\documentclass[prx,aps,twocolumn,notitlepage,superscriptaddress,showpacs,nofootinbib]{revtex4-2}

\usepackage{enumerate,appendix}
\usepackage{amsmath, amsthm, amssymb,commath}
\usepackage{color,calc,graphicx}
\usepackage[usenames,dvipsnames,svgnames,table,cmyk,hyperref]{xcolor}
\usepackage[colorlinks]{hyperref}
\usepackage{optidef}
\hypersetup{
	colorlinks = true,
	urlcolor = {blue},
	citecolor = {blue},
	linkcolor= {blue}
}

\usepackage{graphicx}
\usepackage{amsmath}
\usepackage{latexsym}
\usepackage{bbm}

\usepackage[charter,cal=cmcal,sfscaled=false]{mathdesign}
\usepackage{booktabs}
\usepackage{multirow}
\usepackage{subcaption}
\usepackage{dcolumn}
\usepackage{mathrsfs}
\usepackage[normalem]{ulem}
\usepackage{makecell}
\usepackage{diagbox}

\def \ox{\otimes}

\def \be {\begin{equation}}
\def \ee {\end{equation}}

\newcommand{\Tr}{\mathrm{Tr}}

\newcommand{\ket}[1]{|#1\rangle}
\newcommand{\bra}[1]{\langle#1|}

\newcommand{\kett}[1]{|#1\rangle\!\rangle}
\newcommand{\braa}[1]{\langle\!\langle#1|}

\def \cH{{\cal H}}

\def \sofc2{{\cal S}({\mathbb C}^2)}

\def\>{\rangle}
\def\<{\langle}
\def\opr{\operatorname}

\if0

\newcommand{\ket}[1]{\ensuremath{|#1\rangle}}

\newcommand{\bra}[1]{\ensuremath{\langle#1|}}

\newcommand{\beq}{\begin{equation}}
\newcommand{\eeq}{\end{equation}}
\newcommand{\bqa}{\begin{eqnarray}}
\newcommand{\eqa}{\end{eqnarray}}
\newcommand{\Tr}{\textrm{Tr}}

\newcommand{\forget}[1]{}
\renewcommand{\theequation}{S\arabic{equation}}
\renewcommand{\thefigure}{S\arabic{figure}}

\def \cH{{\mathcal H}}

\fi

\newtheorem{theorem}{Theorem}

\begin{document}

\title{Predicting symmetries of quantum dynamics with optimal samples}		

\begin{abstract}
Identifying symmetries in quantum dynamics, such as identity or time-reversal invariance, is a crucial challenge with profound implications for quantum technologies. We introduce a unified framework combining group representation theory and subgroup hypothesis testing to predict these symmetries with optimal efficiency. By exploiting the inherent symmetry of compact groups and their irreducible representations, we derive an exact characterization of the optimal type-II error (failure probability to detect a symmetry), offering an operational interpretation for the quantum max-relative entropy. In particular, we prove that parallel strategies achieve the same performance as adaptive or indefinite-causal-order protocols, resolving debates about the necessity of complex control sequences. Applications to the singleton group, maximal commutative group, and orthogonal group yield explicit results: for predicting the identity property, Z-symmetry, and T-symmetry of unknown qubit unitaries, with zero type-I error and type-II error bounded by $\delta$, we establish the explicit optimal sample complexity which scales as $\mathcal{O}(\delta^{-1/3})$ for identity testing and $\mathcal{O}(\delta^{-1/2})$ for T/Z-symmetry testing. These findings offer theoretical insights and practical guidelines for efficient unitary property testing and symmetry-driven protocols in quantum information processing.
\end{abstract}

\author{Masahito Hayashi}\email{hmasahito@cuhk.edu.cn}
\affiliation{School of Data Science, The Chinese University of Hong Kong, Shenzhen, Longgang District, Shenzhen, 518172, China}
\affiliation{International Quantum Academy, Futian District, Shenzhen 518048, China}
\affiliation{Graduate School of Mathematics, Nagoya University, Nagoya, 464-8602, Japan}
\author{Yu-Ao Chen}\email{yuaochen@hkust-gz.edu.cn}
\affiliation{Thrust of Artificial Intelligence, Information Hub,\\
The Hong Kong University of Science and Technology (Guangzhou), Guangzhou 511453, China}
\author{Chenghong Zhu}\email{czhu854@connect.hkust-gz.edu.cn}
\affiliation{Thrust of Artificial Intelligence, Information Hub,\\
The Hong Kong University of Science and Technology (Guangzhou), Guangzhou 511453, China}
\author{Xin Wang}\email{felixxinwang@hkust-gz.edu.cn}
\affiliation{Thrust of Artificial Intelligence, Information Hub,\\
The Hong Kong University of Science and Technology (Guangzhou), Guangzhou 511453, China}
\maketitle

\textbf{Introduction.}---
Predicting quantum properties is a fundamental task in quantum information science, designed at efficiently identifying specific characteristics of quantum systems and serving as a crucial step in understanding their symmetries and structural features. This task is crucial across various domains, where distinguishing certain classes of quantum operations or identifying particular state properties allows for targeted manipulation and control, contributing to advancements in quantum computing, many-body physics, cryptography, and communication. 

The standard method for predicting quantum dynamics properties is to apply quantum process tomography~\cite{QCQI, ChuangNielsen, PCP} and then analyze the classical data.  Although many improvements and variants have been proposed~\cite{tomo1, tomo2, tomo3, tomo4, tomo5, tomo6, tomo7, tomo8, tomo9, tomo10, tomo11, tomo12, tomo13, tomo14, tomo15, tomo16}, the resource consumption of this approach still grows quickly as the system becomes large. This inefficiency has motivated the development of predicting properties of quantum dynamics via quantum algorithms~\cite{Wang2011a,Montanaro2013a,Janzing2005,Ji2009,She2023,qamemory3,yuao}, which aim to efficiently identify key features of a quantum process without necessitating a complete classical reconstruction.

A well-known quantum dynamic property predicting task involves determining whether a given operation is the identity unitary, a problem closely linked to the equivalence problem~\cite{Janzing2005,Ji2009} and essential for calibrating quantum systems~\cite{identityCali}.
Other property-specific testing protocols focus on important symmetry properties such as time-reversal symmetry (T-symmetry) unitaries and diagonal unitaries (Z-symmetry). T-symmetry arises from representations of the orthogonal group and highlights the fundamental role of complex numbers in quantum physics. It has significant implications for both theoretical frameworks and experimental aspects of quantum physics~\cite{ExpI1, ExpI2, ExpI3, tsymm1}, including demonstrating quantum advantage through the use of quantum memory~\cite{qamemory1, qamemory2, qamemory3}.
It has further motivated the development of the quantum resource theory of imaginarity~\cite{RTOI1, RTOI2}. On the other hand, Z-symmetry are often considered less resource-intensive in experimental physics and play a crucial role in efficiently generating quantum randomness~\cite{randomness1, randomness2}, simulating classical thermodynamics~\cite{diagonal1}, and enabling effective quantum circuit synthesis~\cite{diagonal2}.

These properties are fundamental in understanding the complex and non-intuitive behavior of quantum systems. To evaluate the capabilities of quantum computers in testing or predicting these features and to determine their fundamental limits, it is essential to optimally and exactly distinguish them from others.
The key metric for this task is the type-II error, which measures the likelihood that a quantum algorithm fails to identify a system with certain properties, corresponding to hypothesis testing between the set of unitaries with those properties and the entire unitary group.

However, determining the fundamental limits of distinguishability and the exact performance in the general unitary hypothesis testing remains largely unexplored, particularly when multiple dynamics are considered simultaneously. Recent work~\cite{yuao} has also considered the hypothesis testing of T-symmetry and Z-symmetry, which established the optimal protocol only for cases with up to 6 queries. The challenge is further compounded by the complexity of searching for optimal testing protocols for quantum dynamic properties, which require not only sequential and adaptive strategies but also the incorporation of indefinite causal order, as the latter has demonstrated substantial advantages in various quantum information processing tasks~\cite{Bavaresco_2021,chiribella2021indefinite,zhao2020quantum,chapeau2021noisy}. It is therefore essential to develop a general framework to gain a deeper understanding of quantum systems.

\begin{figure}[h]
    \centering
    \includegraphics[width=0.9\linewidth]{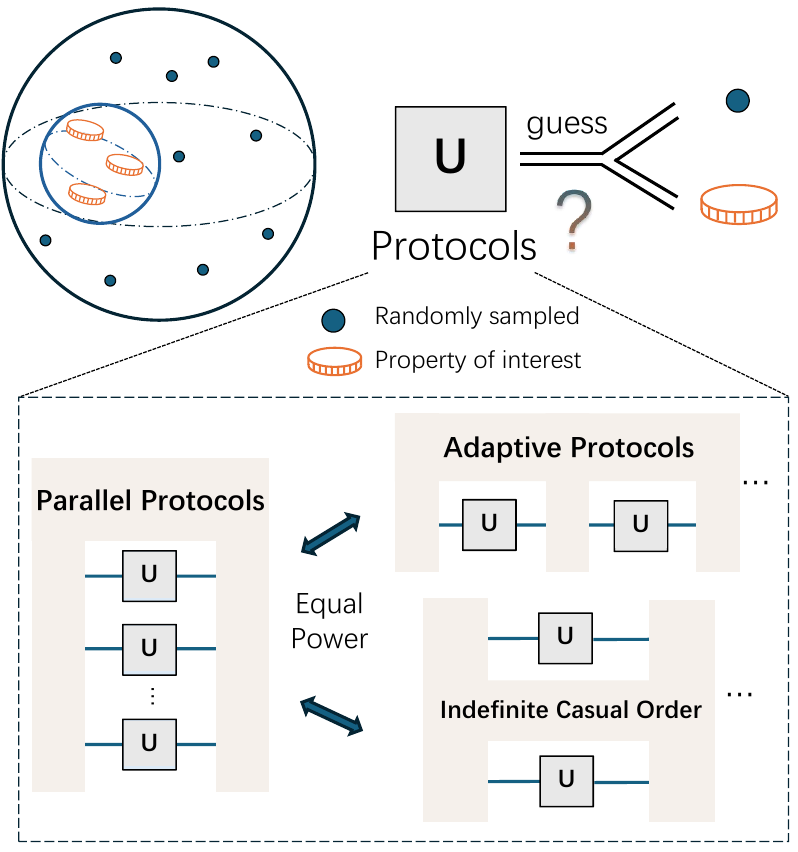}
    \caption{An illustration of hypothesis testing for quantum dynamics of interest, considering protocols such as parallel, adaptive, and indefinite causal order strategies.}
    \label{fig:enter-label}
\end{figure}

In this paper, we towards answer these questions by using the group representation theory. For the property prediction of quantum dynamics task, we demonstrate that parallel strategies can achieve the fundamental limits of the optimal type-II error with adaptive and indefinite-causal-order operations offering no advantage over parallel strategies. We further show the optimal type-II error can also be characterized by the quantum max-relative entropy of the performance operators associated with the two unitary groups, which offers another operational meaning for it. 

As notable applications, we provide the analytical formulas for testing identity, diagonal and real unitaries in the qubit case, characterizing their optimal type-II error scaling behavior. Under zero type-I error and type-II error bounded by $\delta$, we show that the sample complexity of discriminating identity unitary scales as $\mathcal{O}(\delta^{-1/3})$ and it has a natural extension to discriminate an arbitrary unitary $U$, if given access to the $U^{\dagger}$. For the application of testing diagonal and real unitaries, we demonstrate that the sample complexity scales as $\mathcal{O}(\delta^{-1/2})$. We also show that the optimal parallel tests for Z-symmetry and T-symmetry unitaries do not require the ancilla system, whereas testing the identity unitary requires. This distinction provides additional insights into designing protocols for quantum dynamics analysis.

\textbf{Unitary subgroup hypothesis
testing.}--- Let $G$ be a compact group with a subgroup $G_0$ and let $\mu$ and $\mu_{0}$ denote their respective Haar measures. We focus on a unitary representation $f$ of $G$ over a finite-dimensional Hilbert space $\cH$. Under the unitary operation defined by $f$, we study the hypothesis testing problem of distinguishing whether the true unitary is subject to $\mu$ or $\mu_{0}$, including cases where $G$ or $G_0$ is a finite group. This setting is motivated by the fact that, in the absence of any prior information, the unitary operation is naturally modeled as a uniform mixture under the Haar measure, corresponding to the "white noise" operation~\cite{MMH1,MMH2,MMH3,MMH4}. The task then becomes detecting whether a given operation is generated by an element of $G_0$ or is effectively indistinguishable from the white noise operation.

We also introduce the Choi operators of a quantum channel and the performance operator. In detail, the Choi operator of a unitary $U$ is defined as $\kett{U}\braa{U}$, where
\begin{equation}
    \kett{U}\coloneqq\sum_{k,k'}u_{k,k'}\ket{k, k'} \;\text{ for }U=\sum_{k,k'}u_{k,k'}\ket{k}\bra{k'}.
\end{equation}
It is noted that $\kett{U}\braa{U}$ is isomorphic to $U\ox U^\dagger$. Then, the performance operator~\cite{poperator} of a representation $f$ on the Haar measure of group $G$ is then defined as,
\begin{equation}\label{eq:average performance state}
    \rho_{\mu}^f \coloneqq\mathbb E_{U\sim\mu}\kett{f(U)}\braa{f(U)}.
\end{equation}
Similarly, the performance operator for the group $G_0$ is represented by $\rho_{\mu_0}^f$.

% \textbf{Unitary subgroup hypothesis
% testing.--} 
For predicting symmetries of quantum dynamics task, our goal is to use a tester or the measurement operator to distinguish an element of the subgroup $G_0$ from the white noise operation
$\mu$. We then define the measurement operator $T_G$ associated with $\mu$ and $T_{G_0}$ associated with $G_{0}$, as positive semi-definite operators acting on the system $\mathcal H^{\ox 2}$. 

Given the importance of accurately identifying the properties of quantum dynamics and avoiding incorrect predictions, we adopt the best practices of hypothesis testing in quantum information theory. Specifically, it is natural to consider the worst-case analysis of type-I error, the largest failure probability to detect an element $g \in G_0$, as follows:
\begin{equation}
    \alpha(T_{G_0}, f) = 1-\min_{g \in G_0} \Tr (T_{G_0} |f(g)\rangle\!\rangle \langle\!\langle f(g)|).
\end{equation}
The reason is to ensure that the protocol performs reliably for the unitaries of interest in the worst-case scenario, guaranteeing reliable results even under the most challenging cases. 
It is a standard method in hypothesis testing to consider 
the worst case probability when the hypothesis to be detected is composed of multiple elements \cite{BP,HY,Lami}.
We also note there is another work that considers the average error~\cite{yuao} and we will show the equivalence of both worst-case and average-case analysis. Accordingly, we define the average-case type-I and type-II errors as follows:
\begin{equation}
\begin{aligned}
    &\Bar{\alpha}(T_{G_0}, f) = 1 - \Tr (T_{G_0} \rho^f_{\mu_0}), \; \Bar{\beta}(T_{G_0}, f) = \Tr (T_{G_0} \rho^f_{\mu}).
\end{aligned}
\end{equation}
Consider the scenario where the true unitary governed by $\mu_{0}$ must be identified within an error tolerance of $\epsilon$, then the problem can be formalized as the asymmetric hypothesis testing, which asks about how small one of the errors can be subject to constraints on the other error. To account for the search for optimal strategies, we consider scenarios where the operator or protocols $T_G+T_{G_0}$ must satisfy the indefinite-causal-order (ICO) strategy linear conditions~\cite{Bavaresco_2021}. The detailed formulation for both worst-case and average-case analysis are given by,
\begin{equation}
    {\beta}^{f}_{ICO}(\epsilon) :=  \min_{T_G+T_{G_0}:ICO} \{ 
\Bar{\beta}(T_{G_0}, f) :
\alpha(T_{G_0}, f) \le \epsilon \} .
\end{equation}
\begin{equation}
    \Bar{\beta}^{f}_{ICO}(\epsilon) :=  \min_{T_G+T_{G_0}:ICO} \{ 
\Bar{\beta}(T_{G_0}, f) :
\Bar{\alpha}(T_{G_0}, f) \le \epsilon \} .
\end{equation}

On the other side, we also examine the case of parallel protocols. That is, an intuitive approach involves introducing a reference system $\mathcal A$, choosing a pure input state $\ket{\psi}\in\mathcal H\ox\mathcal A$, operating $f(g)_{\cH}$ on $\ket{\psi}$ and choosing a measurement $M\subset\mathcal L^\dagger(\mathcal H\ox\mathcal A)$ on the final state. For this case, one can define,
\begin{equation}
    {\beta}^{f}_{PAR}(\epsilon) :=\inf_{\mathcal A}\min_{\tiny \begin{array}{c}
         |\psi\rangle\in {\mathcal H}\otimes{\mathcal A}\\
         0\preceq M\preceq I_{\mathcal H\ox\mathcal A} 
    \end{array}} \{ 
\Bar{\beta}(T_{G_0}, f) :
\alpha(T_{G_0}, f) \le \epsilon \},
\end{equation}
where the formalized tester $T_{G_0}$ is identified by $\ket{\psi}$ and $M$.

\textbf{Main Results.}--- Next, we present the main theorem. We begin by examining the inclusion of different operational strategies, then establish a connection between the type-II error and the max-relative entropy, and finally provide the explicit calculation of the type-II error. The formal statement is presented as follows,
\begin{theorem}[Optimal type-II error of unitary subgroup hypothesis testing]\label{thm:theorem1}
For unitary subgroup hypothesis testing involving compact group $G$ and its subgroup $K$ with unitary representation $f$ and an error tolerance of $\epsilon$, the following holds,
\begin{equation}
    \beta^{f}_{PAR}(\epsilon)
    = \beta^{f}_{ICO}(\epsilon) 
    = \Bar{\beta}^{f}_{ICO}(\epsilon) 
    = (1-\epsilon)e^{-D_{\max} (\rho^f_{\mu_0}\|\rho^f_{\mu})}.
\end{equation}
where $\rho^f_{\mu_0}$ and $\rho^f_{\mu}$ are performance operator of a representation $f$ on the Haar measure of group $G_0$ and $G$, respectively (cf.~Eq.~\eqref{eq:average performance state}).
\end{theorem}
The theorem has several implications, which we will detail individually as follows. We also note the detailed proof is provided in Section \ref{sec: appendix main proof} of the Supplementary Material. 
\paragraph{Parallel Protocols Achieve Optimality.} Our findings rigorously demonstrate that parallel protocols, even in their most straightforward implementation with pure state inputs are sufficient to achieve the optimal type-II error in the task of unitary subgroup hypothesis testing. The simplicity of parallel protocols makes them a particularly practical and efficient approach to hypothesis testing of properties of quantum dynamics.

Notably, our results also reveal that sequential and other causal strategies do not provide any advantage over parallel protocols in this context. The parallel protocols are shown to be as powerful as indefinite causal order strategies, establishing their equivalence in effectiveness. This conclusion serves as a counterexample to previous assertions in areas such as channel discrimination~\cite{Bavaresco_2021}, quantum communication~\cite{chiribella2021indefinite}, and quantum metrology~\cite{zhao2020quantum, chapeau2021noisy}, where indefinite causal order or adaptive strategies were suggested to offer a distinct advantage.

\paragraph{Connection to Max-Relative Entropy.} The second implication of the main theorem is the connection of the average type-II error with the quantum max-relative entropy~\cite{Datta2009}, which can be understanded as the sandwiched R\'enyi divergence of order $\infty$. The definition of the quantum max-relative entropy is given by $D_{\max}(P\| Q) := \min\{t: e^tQ \geq P\}$, where $P$ and $Q$ are positive semi-definite operators. 

The optimal type-II error with an error tolerance of $\epsilon$ can be determined through a direct calculation of $D_{\max} (\rho^f_{\mu_0}\|\rho^f_{\mu})$. The quantum max-relative entropy has found numerous applications in quantum information processing tasks~\cite{Datta2009, Bu2017, Wang2020, Zhu2024}, and here we provide it with a new operational interpretation by directly linking it to the task of quantifying the average type-II error in hypothesis testing of quantum dynamics. We also observe that the quantum max-relative entropy represents a further relaxation of all causal order strategies, as it only requires the tester to be a positive semidefinite operator without the need to satisfy additional linear conditions imposed by adaptive or IDC strategies.

\paragraph{Equivalence of Average and Worst-Case Analysis.} The theorem also establishes that hypothesis testing for quantum dynamics exhibits an equivalence between the worst-case and average-case scenarios. Specifically, the optimal performance of hypothesis testing protocols, typically characterized by type-II errors, remains invariant whether evaluated over all possible instances (worst-case) or averaged with respect to the distribution of interest. The similar phenomenon is also well-established for models with certain group covariance structures, as described by the Hunt-Stein theorem~\cite{Hol-Cov, Hol-Pro, Group2}.

\paragraph{Error Tolerance.} The theorem further establishes a linear dependence of the optimal type-II error, i.e. $\beta_{PAR}^f (\epsilon) = (1-\epsilon)\beta_{PAR}^f (0)$.  This result indicates that the protocols achieving the optimal type-II error remain optimal even when a non-zero tolerance $\epsilon$ is allowed for the type-I error. Consequently, for the later application section, we focus exclusively on the case where $\epsilon=0$.

Following Theorem~\ref{thm:theorem1}, we now discuss the method for evaluating the optimal type-II error. By introducing unitary representations of compact groups, we have the following theorem for the calculation of optimal type-II error,
\begin{theorem}[General solutions]~\label{thm:2}
The optimal type-II error for unitary subgroup hypothesis testing of unitaries draw from $\mu$ and $\mu_{0}$ is given by,
\begin{equation}\label{eq:th2}
    e^{-D_{\max} (\rho^f_{\mu_0}\|\rho^f_{\mu})} = \min_{\eta\in\hat{G_0}_f}\frac{d_{\eta, G_0}}{\sum_{\lambda\in\hat G_f} d_\lambda n_{\eta, \lambda}},
\end{equation}
where $\hat G_f$ and $\hat{G_0}_f$ denote the set of irreducible representation of group $G$ and $G_0$ appears in $f$, respectively;
$d_\lambda$ and $d_{\eta, G_0}$ denote the dimensions of representation $\lambda$ and $\eta$, respectively;
for fixed $\lambda,\eta$, $n_{\eta,\lambda}$ denotes the multiplicity of irreducible representation $\eta$ of group $G_0$ appears in an irreducible representation space corresponding to irreducible representation $\lambda$ of group $G$.
\end{theorem}
Theorem~\ref{thm:2} provides a general solution for evaluating the optimal type-II error $\beta^{f}_{PAR}(\epsilon)$ by adapting 
Eq.~\eqref{eq:th2} to Theorem \ref{thm:theorem1}. This solution depends solely on the irreducible representations of the groups $G$ and $G_0$. This formula establishes a connection between the type-II error and the group structure, offering a systematic framework for analyzing hypothesis testing problems. The detailed proof can be found in Supplementary Material~\cite{SM}. 
This result also implies that when $G_0$ is a unitary $n$-design, it is impossible to distinguish between $U\sim\mu_{G_0}$ and $U\sim\mu_{\opr U(d)}$ within $n$ uses of $U$. Below, we will demonstrate how to apply it in the calculation of the type-II error associated with the fundamental properties of interest.

\textbf{Applications.}--- After introducing the main technical framework, we now discuss the situation when we considering the trivial subgroup $G_0 = \{I\}$, maximal commutative group $G_0\cong\operatorname{U}(1)^{d}$ and orthogonal group $G_0 =\operatorname{O}(d)$ by using the $n$-fold tensor product representation $f_n(g):=g^{\otimes n}$. These have applications in distinguishing the identity, diagonal unitaries as well as real unitaries from $d$-dimensional unitary group $G=\operatorname{U}(d)$. We then derive analytical formulas for the optimal type-II error as a function of the number of queries $n$ to the unknown unitaries.

\textbf{Identity Testing.} The first application we consider is the identity test, which involves determining whether an unknown quantum operation is equivalent to the identity transformation. This task is crucial for verifying and calibrating quantum systems~\cite{identityCali} and it serves as a foundation for more complex property testing tasks, such as distinguishing trivial operations from non-trivial ones, and plays a vital role in benchmarking quantum devices through protocols like randomized benchmarking.

In a technical detail, we consider the case where $G=\opr{U}(d)$ and the subgroup only contains a single element $G_0 = \{I\}$. Since $ G_0 = \{I\}$ is a trivial group, it possesses only the trivial irreducible representation, which implies $d_\eta=1$ and $\forall\,\lambda\in\hat G_f, n_{\eta,\lambda}=d_\lambda$. 
Then one can show that the optimal average type-II error is,
\begin{equation}
\beta_{PAR}^f(0) = 1/\sum_{\lambda\in\hat G_f} d_\lambda^2.
\end{equation}
Then, following the calculations presented in \cite{Population}, the analytical formulas for the optimal average type-II error when $d=2$ can be derived as follows,
\begin{equation}
\beta_{PAR}^{f_n}(0) = \frac{6}{(n+1)(n+2)(n+3)},
\end{equation}
where $n$ represents the number of queries to the unknown unitaries. We remain the detailed derivation in the Supplementary Material~\cite{SM}.

We also remark that the identity test can be extended to testing arbitrary quantum dynamics. If access to $U^\dagger$ is available, the key idea is to apply $U^\dagger$ to the unknown dynamics, effectively reducing the problem back to the identity test. The preparation of the reverse quantum dynamics can also leverage deterministic and exact protocols, as demonstrated in~\cite{reverseunknownd2, reverseunknownda}. These protocols do not require full knowledge of the process and are applicable to any unitary operation, making them broadly applicable and efficient for this task.

\textbf{Z-symmetry Testing.}
Besides the identity test, we will discuss its application in the Z-symmetry prediction. The Z-symmetry prediction is a fundamental tool in quantum information theory, used to identify and distinguish Z-symmetry from more general quantum dynamics. This has significant implications for various quantum tasks, as Z-symmetry play a crucial role in encoding phase information and preserving certain symmetries in quantum systems. The ability to efficiently test for Z-symmetry is essential in resource theories, where these operations often correspond to conserved resources such as coherence or asymmetry~\cite{RTOI1, RTOI2}. Furthermore, the Z-symmetry test is vital for optimizing and verifying quantum algorithms~\cite{diagonal2}.

In the technical detail, we consider the case $G=\opr U(d)$, the maximal commutative subgroup $G_0 \cong\opr U(1)^d$ and $n$-fold tensor product representation $f_n(g) := g^{\ox n}$. Then for the qubit case, we show that when the number of queries is $n$, the optimal average type-II error satisfies the following formula,
\begin{equation}
\beta_{PAR}^{f_n}(0) = 
\begin{cases}
    \displaystyle \frac{4}{(n+2)^2}, &\textrm{ if } n \textrm{ is even}; \vspace{1em}\\
    \displaystyle \frac{4}{(n+1)(n+3)}, &\textrm{ if } n \textrm{ is odd}. 
\end{cases}
\end{equation}
We remain the detailed derivation in the Supplementary Material~\cite{SM}.
This overall shows a $\mathcal{O}(n^{-2})$ scaling of optimal type-II error decay for any number of queries $n$, which solves the problem in the previous work~\cite{yuao} which they are only limited to prove the result for limited number of queries, i.e. 1-5 queries. Unlike the identity test, we demonstrate that achieving the optimal type-II error does not require a reference system.

\textbf{T-symmetry Testing.}
As a final application, we will consider the predication of T-symmetry unitaries. T-symmetry unitaries are critical in quantum mechanics as they represent operations that are invariant under time reversal and they correspond to real orthogonal matrices when expressed in specific bases and often arise in physical systems with intrinsic symmetries. Identifying T-symmetry unitaries is particularly important for understanding time-reversal invariance in quantum dynamics, a property that plays a central role in quantum thermodynamics, condensed matter physics, and symmetry-protected topological phases~\cite{ts_topo1, ts_topo2}.

With the similar spirit to the identity and Z-symmetry testing, we consider the case with $G = \opr U(d)$ and its subgroup $ G_0 = \opr O(d)$. 
We show that when the number of queries is $n$, the optimal average type-II error for $d=2$,
\begin{equation}
\beta_{PAR}^{f_n}(0) = 
\begin{cases}
    \displaystyle \frac{8}{(n+2)(n+4)}, &\textrm{ if } n \textrm{ is even}; \vspace{1em}\\
    \displaystyle \frac{8}{(n+1)(n+3)}, &\textrm{ if } n \textrm{ is odd}. 
\end{cases}
\end{equation}
We remain the detailed derivation in the Supplementary Material~\cite{SM}. We remark that for T-symmetry testing, the optimal type-II error does not improve when increasing the number of queries from an even $n$ to the next odd number of queries $n+1$, i.e. $\beta_{PAR}^{f_{2n+1}}(0) = \beta_{PAR}^{f_{2n}}(0)$. This result highlights a symmetry in the performance of the testing protocols, where the additional query in the odd case does not contribute to further reducing the error.

To summarize, we present three applications: the identification of identity, Z-symmetry, and 
T-symmetry unitaries, and plot their corresponding error decay as a function of the number of queries to the unknown unitaries in Fig.~\ref{fig:application_scaling}. The results demonstrate that identifying the identity unitary is easier compared to Z-symmetry and T-symmetry unitaries. This is intuitively understandable, as the identity test involves distinguishing a single element, making it simpler than the Z-symmetry and T-symmetry cases. In all scenarios, the type-II error rapidly approaches negligible values as the number of queries increases.
\begin{figure}[t]
    \centering
    \includegraphics[width=0.95\linewidth]{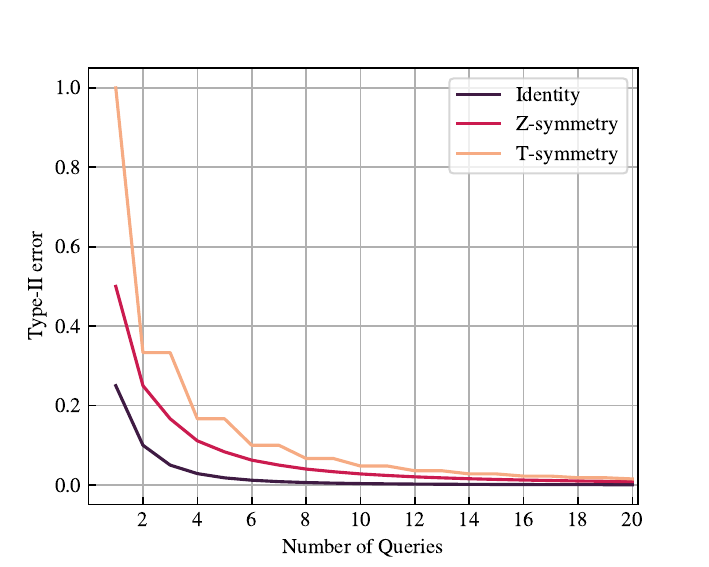}
    \caption{The optimal type-II error scaling with respect to the number of queries in identity, Z-symmetry and T-symmetry identification cases.}
    \label{fig:application_scaling}
\end{figure}

\textbf{Concluding remarks.}--- We have developed a framework for predicting properties of unknown quantum dynamics via unitary subgroup hypothesis testing. We have shown the equivalence between worst-case and average-case analysis and link the optimal type-II error to the quantum max-relative entropy. We have established the analytical formulas for precisely characterizing the optimal type-II error in applications such as identity, Z-symmetry and T-symmetry testing.  For these tasks, under zero type-I error and a type-II error bounded by $\delta$, the optimal sample complexity is $\mathcal{O}(\delta^{-1/3})$ for identity testing and $\mathcal{O}(\delta^{-1/2})$ for Z-symmetry and T-symmetry testing.
Furthermore, we have proved that all optimal type-II errors can be achieved using parallel protocols. This overturns assumptions about the necessity of complex control sequences, simplifying the testing process for quantum advantage experiments. 
% paving the way for scalable approaches to analyzing quantum dynamics.

This paper also shows that adaptive operation or indefinite casual order strategies does not improve the performance of hypothesis testing for group action. In fact, similar facts have been proved in the contest of secure network coding~\cite{HC1,HOKC,CH,H24,KOH,OKH,HS}, 
the discrimination of elements in finite group~\cite{Bavaresco2022},
and the estimation of group action~\cite{HayashiEstimation}. 
The preceding studies show that adaptive modifications of the input information in each attacked edge by the adversary does not improve the information gain by the adversary when all coding operations are given as linear operations for classical secure network coding~\cite{HC1,HOKC,CH,H24} and quantum secure network coding~\cite{KOH,OKH,HS}. But, Ref.~\cite{HC2} highlights that such improvement exists when nonlinear network code is applied. An intriguing avenue for future research is to investigate whether this phenomenon is fundamentally tied to group symmetry.

Another interesting direction is to explore a more general and efficient approach for evaluating the type-II error for unitaries of arbitrary dimensions.
Future research could also explore the hypothesis testing framework when the subgroup is chosen as the Weyl-Heisenberg group or the Clifford group. Similarly, when $G$ is the Clifford group, $G_0$ can be considered as the Weyl-Heisenberg group. 

\textbf{Acknowledgments.--} 
M. H. was supported in part by the National Natural Science Foundation of China (Grant No.~62171212). 
Y.-A. C., C. Z., and X. W. were partially supported by the National Key R\&D Program of China (Grant No.~2024YFE0102500), the National Natural Science Foundation of China (Grant. No.~12447107), the Guangdong Provincial Quantum Science Strategic Initiative (Grant No.~GDZX2403008, GDZX2403001), and the Guangdong Provincial Key Lab of Integrated Communication, Sensing and Computation for Ubiquitous Internet of Things (Grant No.~2023B1212010007).

% %%%%%%%%% -- BIB STYLE AND FILE -- %%%%%%%%
% \bibliographystyle{ieeetr}
% \bibliography{refs}

\clearpage
% \appendix
\widetext

\begin{center}
\textbf{\large Supplemental Material for \\``Predicting symmetries of quantum dynamics with optimal samples''}
\end{center}

\renewcommand{\thedefinition}{S\arabic{definition}}
\setcounter{definition}{0}

% \numberwithin{equation}{section}
\renewcommand{\thefigure}{S\arabic{figure}}
\setcounter{figure}{0}
\renewcommand{\theequation}{S\arabic{equation}}
\setcounter{equation}{0}
\renewcommand{\thesection}{\Roman{section}}
\setcounter{section}{0}
\setcounter{secnumdepth}{4}
% use ``secnumdepth'' to show up the section number

In the supplemental material, we first introduce some preliminaries and elaborate on the basic setup of this work. Next, we provide the rigorous proof for the theorem presented in the main text.

\begin{itemize}
\item The proof of Theorem \ref{thm:theorem1} in the main text is provided in Section \ref{appendix:thm_3}, \ref{sec:KH1}, and \ref{sec:KH2}. The theorem statements are presented in Section \ref{appendix:thm_3} Theorem \ref{NMI}, while the main proof is broken down into Sections \ref{sec:KH1} and \ref{sec:KH2}.
\item The statement and proof of Theorem \ref{thm:2} in the maxin text are provided in Section \ref{appendix:cal_max_relative} Theorem \ref{appendix:THM4}.
\item Lastly, we present the detailed calculations for the three different applications in Section~\ref{sec:appendix case study}.
\end{itemize}

\section{Notations}
Let $G$ be a compact group, and $G_0$ be its subgroup.
Also, let $\mu$ be the Haar measure of $G$ and $\mu_0$ be the Haar measure of $G_0$.
We focus on a unitary representation $f$ of $G$ over a 
finite-dimensional Hilbert space.
Under the unitary operation by $f$,
we study the hypothesis testing for whether the true unitary 
is subject to $\mu$ or one of $G_0$.
This setting includes the case when $G_0$ (and $G$) is a finite group.
For this aim, we choose an input state and a binary measurement for our decision.
We assume that any entangled input state with any reference system is available.

To consider this problem, we denote the set of the labels of 
irreducible representations of $G$ by $\hat{G}$.
Let ${\cal U}_\lambda$ be the irreducible representation space 
identified by $\lambda \in \hat{G}$, and
$d_\lambda$ be its dimension.
We also define the twirling operation for the group $G$ as
\begin{align}
{\cal T}_G(\rho):=\int_{G}
 f(g')\rho f(g')^\dagger \mu(dg').
\end{align}
When we focus on a subgroup $G_0$, we define 
${\cal T}_{G_0}$ by using the Haar measure on $G_0$ in the same way.

\section{Detailed proofs of the main results}\label{sec: appendix main proof}
In this section, we will provide the detailed proof of the main results.

\subsection{General problem formulation and Statement of Theorem \ref{NMI} }\label{appendix:thm_3}

We consider $n$ representation spaces ${\cal H}_1, \ldots, {\cal H}_n$ 
of $G$.
We consider each representation of ${\cal H}_j$
as a channel with the input system ${\cal H}_j$
and the output system ${\cal H}_j$ for $j=1, \ldots,n$.
We consider indefinite-causal-order measurement on 
these $n$ channels.
We describe each channel, i.e., each unitary action by 
$d$ times of its Choi representation, where
$d$ is the dimension of the whole output system.
We denote this representation of the unitary $U$ 
by $|U\rangle\rangle\langle \langle U|$.
Here, when $U$ is a matrix $\sum_{k,k'}u_{k,k'} |k\rangle \langle k'|$,
$|U\rangle\rangle$ is defined as
\begin{align}
|U\rangle\rangle:= \sum_{k,k'} u_{k,k'} |k,k'\rangle.
\end{align}
For example, when ${\cal H}={\cal H}_1 \oplus {\cal H}_2$, we have
\begin{align}
|I_{\cal H}\rangle\rangle= 
|I_{{\cal H}_1}\rangle\rangle \oplus |I_{{\cal H}_2}\rangle\rangle. \label{NJ1}
\end{align}
When ${\cal H}={\cal H}_1 \otimes {\cal H}_2$, we have
\begin{align}
|I_{\cal H}\rangle\rangle= 
|I_{{\cal H}_1}\rangle\rangle \otimes |I_{{\cal H}_2}\rangle\rangle. \label{NJ2}
\end{align}

A measurement operator $T_k$ is given as a positive semi-definite operator 
over the system $\otimes_{j=1}^n {\cal H}_j^{\otimes 2}$.
For the set $\{T_k\}$,
the operator $\sum_{k}T_k$ needs to satisfy the 
Indefinite-Causal-Order Strategy linear condition, ICO-condition, whose detail is given in \cite[Eqs. (4),(5),(6)]{Bavaresco_2021}.

% \textcolor{blue}{(Chenghong: This sentence can be moved, as discussed in the main text.)} \CH{Generally, when a model has a certain group covariance structure,
% it is known that the worst case analysis is given as the analysis on the average case as Hunt-Stein theorem \cite[Theorem 2]{Hol-Cov}, \cite[Theorem 4.3.1]{Hol-Pro}, \cite[Theorem 4.1]{Group2}.}

We consider a test $T=\{T_0,T_1\}$, where 
$T_0$ supports an operation of the subgroup $G_0$ 
and
$T_1$ supports the Haar measure on $G$. 
$T_0+T_1$ satisfies the ICO-condition.

Now, the whole system is 
$(\otimes_{j=1}^n {\cal H}_j)\otimes (\otimes_{j=1}^n {\cal H}_j)$,
the group $G$ acts only on the first system $(\otimes_{j=1}^n {\cal H}_j)$, and the second system
${\cal K}:=(\otimes_{j=1}^n {\cal H}_j)$ can be considered as multiplicity space.
We denote the dimension of ${\cal H}$ by $d$.
So, $\dim{\cal K}=d$. 

We denote the representation on the first system 
${\cal H}:=(\otimes_{j=1}^n {\cal H}_j)$ as
\begin{align}
{\cal H}= \bigoplus_{\lambda \in \hat{G}_f}
{\cal U}_\lambda \otimes \mathbb{C}^{n_\lambda},
\end{align}
where 
$n_\lambda$ expresses the multiplicity of the representation 
space ${\cal U}_\lambda$.
In the following, we denote the dimension of irreducible space
${\cal U}_\lambda$ by $d_\lambda$.
Here, we denote the representation on ${\cal H}$ by $f$.
We denote the set of irreducible representations appearing in $f$
by $\hat{G}_f$.
An operator $T_0$ is called invariant for $G_0$ when 
\begin{align}
f(g) T_0 f(g)^\dagger = T_0
\end{align}
for any $g \in G_0$.
Then, the whose system is written as
\begin{align}
{\cal H}\otimes {\cal K}
= 
\Big(\bigoplus_{\lambda \in \hat{G}_f}
{\cal U}_\lambda \otimes \mathbb{C}^{n_\lambda} \Big)
\otimes 
\Big(\bigoplus_{\lambda' \in \hat{G}_f}
\mathbb{C}^{d_{\lambda'}} \otimes \mathbb{C}^{n_{\lambda'}} \Big)
\end{align}
By considering \eqref{NJ1} and \eqref{NJ2},
the application of $g$ is written as
\begin{align}
|f(g)\rangle\rangle=&
f(g)|I_{\cal H}\rangle\rangle
=f (g)
 \bigoplus_{\lambda \in \hat{G}_f}
| I_{\lambda}\rangle\rangle \otimes | I_{n_\lambda} \rangle\rangle \\
=&
 \bigoplus_{\lambda \in \hat{G}_f}
f_\lambda (g)| I_{\lambda}\rangle\rangle \otimes | I_{n_\lambda} \rangle\rangle \\
=&
 \bigoplus_{\lambda \in \hat{G}_f}
| f_\lambda (g) \rangle\rangle \otimes | I_{n_\lambda} \rangle\rangle.
\end{align}
The average with respect to the Haar measure $\mu$ is denoted as $\rho_{\mu}^f$ in the main text. For simplicity, we denote it as $\rho_\mu$, which is defined as,
\begin{align}
\rho_{\mu}:=&\int_{G}|f(g)\rangle\rangle \langle\langle f(g)|\mu(dg)
=\bigoplus_{\lambda \in \hat{G}_f}
d_\lambda^{-1} I_\lambda \otimes I_\lambda \otimes 
| I_{n_\lambda} \rangle\rangle\langle\langle I_{n_\lambda}|.
\end{align}
To show the above relation, we 
denote the matrix component of $g_\lambda (g)$
as 
$g_{\lambda, k,l} (g) $.
It is known that 
the set $\{g_{\lambda, k,l} (g) \}_{k,l,\lambda}$
forms a completely orthogonal basis on the $L^2$ function space on $G$ under the inner product
 $\langle h_1,h_2\rangle:= 
 \int_{G} h_1(g) \overline{h}_2(g) \mu(dg)$ \cite[Lemma 2.14]{Group}. 
That is, \cite[Lemma 2.14]{Group} implies the above relation.

Similarly, the average with respect to the Haar measure $\mu_0$ is 
\begin{align}
\rho_{\mu_0}
:=&\int_{G_0}|f(g)\rangle\rangle \langle\langle f(g)|\mu_0(dg) .
%=& {\cal T}_{\mu_0}(|f(g)\rangle\rangle \langle\langle f(g)|)
\end{align}
We define
\begin{align}
D_{\max} (\rho_{\mu_0}\|\rho_{\mu}):= \min
\{t\,|\, 
e^{t}\rho_{\mu} \ge \rho_{\mu_0} \}.
\end{align}
Therefore, our problem is the calculation of
$$
\min_{T:ICO} \{ 
\Tr T \rho_{\mu} :
\min_{g \in G_0} \Tr T |f(g)\rangle\rangle \langle\langle f(g)|
\ge 1-\epsilon \} .$$
To consider this problem, we have the following theorem,

\begin{theorem}[Optimal type-II error of unitary subgroup hypothesis testing]\label{NMI}
For unitary subgroup hypothesis testing involving compact group $G$ and its subgroup $G_0$ with unitary representation $f$ and an error tolerance of $\epsilon$, the following holds,
\begin{align}
&\min_{|\psi\rangle\in {\cal H}\otimes{\cal K}} \min_{T:inv} \{ 
\Tr T {\cal T}_G( |\psi\rangle\langle \psi|) : \min_{g \in G_0} \Tr T f(g) |\psi\rangle\langle \psi| f(g)^\dagger
\ge 1-\epsilon \} \\
=&\min_{T:ICO,inv} \{ 
\Tr T \rho_{\mu} :
\min_{g \in G_0} \Tr T |f(g)\rangle\rangle \langle\langle f(g)|
\ge 1-\epsilon \} \\
=&\min_{T:ICO} \{ 
\Tr T \rho_{\mu} :
\min_{g \in G_0} \Tr T |f(g)\rangle\rangle \langle\langle f(g)|
\ge 1-\epsilon \} \\
=&\min_{T:ICO} \{ 
\Tr T \rho_{\mu} :
 \Tr T \rho_{\mu_0}
\ge 1-\epsilon \} \\
=&\min_{T\ge 0} \{ 
\Tr T \rho_{\mu} :
 \Tr T \rho_{\mu_0}
\ge 1-\epsilon \} \\
=&(1-\epsilon) e^{-D_{\max} (\rho_{\mu_0}\|\rho_{\mu})}.
\end{align}
\end{theorem}

We remark that this theorem still holds even when the ICO-condition is replaced by another linear condition.  For example, this theorem still holds when only Eqs. (4),(5) of \cite{Bavaresco_2021} is imposed instead of Eqs. (4),(5),(6) of \cite{Bavaresco_2021} for the two-system case.

It is first easy to observe that,
\begin{align}
&\min_{|\psi\rangle\in {\cal H}\otimes{\cal K}} \min_{T:inv} \{ 
\Tr T {\cal T}_G( |\psi\rangle\langle \psi|) :
\min_{g \in G_0} \Tr T f(g) |\psi\rangle\langle \psi| f(g)^\dagger
\ge 1-\epsilon \} \notag\\
\ge &\min_{T:ICO,inv} \{ 
\Tr T \rho_{\mu} :
\min_{g \in G_0} \Tr T |f(g)\rangle\rangle \langle\langle f(g)|
\ge 1-\epsilon \} \notag\\
\ge &\min_{T:ICO} \{ 
\Tr T \rho_{\mu} :
\min_{g \in G_0} \Tr T |f(g)\rangle\rangle \langle\langle f(g)|
\ge 1-\epsilon \} \notag\\
\ge &\min_{T:ICO} \{ 
\Tr T \rho_{\mu} :
 \Tr T \rho_{\mu_0}
\ge 1-\epsilon \}  \notag\\
\ge &\min_{T\ge 0} \{ 
\Tr T \rho_{\mu} :
 \Tr T \rho_{\mu_0} \ge 1-\epsilon \} ,
\end{align}
Then, Theorem \ref{NMI} can be shown by proving the following two parts,
\begin{align}
\min_{T\ge 0} \{ 
\Tr T \rho_{\mu} :
 \Tr T \rho_{\mu_0}
\ge 1-\epsilon \} 
=(1-\epsilon) e^{-D_{\max} (\rho_{\mu_0}\|\rho_{\mu})}.\label{KH1}
\end{align}
and
\begin{align}
\min_{|\psi\rangle\in {\cal H}\otimes{\cal K}} \min_{T:inv} \{ 
\Tr T {\cal T}_G( |\psi\rangle\langle \psi|) : \min_{g \in G_0} \Tr T f(g) |\psi\rangle\langle \psi| f(g)^\dagger
\ge 1-\epsilon \} 
\le (1-\epsilon) e^{-D_{\max} (\rho_{\mu_0}\|\rho_{\mu})}.\label{KH2}
\end{align}
We will detail the proof for Eq.\eqref{KH1} and Eq.\eqref{KH2} in Section \ref{sec:KH1} and Section \ref{sec:KH2}, respectively.

\subsection{Proof of \eqref{KH1}}\label{sec:KH1}
For any $T\ge 0$, we have 
\begin{align}
e^{D_{\max} (\rho_{\mu_0}\|\rho_{\mu})}
\Tr T \rho_{\mu} \ge \Tr T \rho_{\mu_0} .
\end{align}
Hence, we have
\begin{align}
\min_{T\ge 0} \{ 
\Tr T \rho_{\mu} :
 \Tr T \rho_{\mu_0}
\ge 1-\epsilon \} 
\ge (1-\epsilon) e^{-D_{\max} (\rho_{\mu_0}\|\rho_{\mu})}.\label{BNR}
\end{align}
Also, we choose a vector $|w\rangle$ such that
\begin{align}
e^{D_{\max} (\rho_{\mu_0}\|\rho_{\mu})}
\langle w| \rho_{\mu}|w\rangle = \langle w| \rho_{\mu_0} |w\rangle.
\end{align}
The operator $T=
\frac{1-\epsilon}{\langle w| \rho_{\mu_0} |w\rangle}
|w\rangle\langle w| $ achieves the equality in \eqref{BNR}.

\subsection{Calculation of $D_{\max} (\rho_{\mu_0}\|\rho_{\mu})$}\label{appendix:cal_max_relative}
We consider the irreducible unitary representations of $G_0$.
We denote the irreducible unitary representation space identified by
 $\eta \in \hat{G}_0$ by
 ${\cal U}_{\eta,G_0}$.
In the following, we denote the dimension of irreducible space
${\cal U}_{\eta,G_0}$ by $d_{\eta,G_0}$.
We consider the irreducible decomposition of ${\cal U}_{\lambda}$ as the representation of $G_0$ as
\begin{align}
{\cal U}_{\lambda}=
\bigoplus_{\eta\in\hat{G}_{0,\lambda}}
{\cal U}_{\eta,G_0}\otimes \mathbb{C}^{n_{\eta,\lambda}},
\end{align}
where $\hat{G}_{0,\lambda}$ is the set of irreducible unitary representations of $G_0$ that appears in ${\cal U}_{\lambda}$
We denote the dimension of ${\cal U}_{\eta,G_0}$ by $d_{\eta,G_0}$.

The state on
${\cal U}_{\lambda}\otimes \mathbb{C}^{d_{\lambda}}$
is written as
\begin{align}
|f_\lambda (g)\rangle \rangle
=
\bigoplus_{\eta\in\hat{G}_{0,\lambda}}
|f_{\eta,G_0} (g)\rangle \rangle
\otimes |I_{n_{\eta,\lambda}}\rangle \rangle
\end{align}

The application of $g$ is written as
\begin{align}
|f(g)\rangle\rangle:
= &
\bigoplus_{\lambda \in \hat{G}_f}
\bigoplus_{\eta\in\hat{G}_{0,\lambda}}
|f_{\eta,G_0} (g)\rangle \rangle
\otimes |I_{n_{\eta,\lambda}}\rangle \rangle
\otimes | I_{n_\lambda} \rangle\rangle \notag\\
= &
\bigoplus_{\eta}
|f_{\eta,G_0} (g)\rangle \rangle \otimes
\Big(\bigoplus_{\lambda \in \hat{G}_f}
 |I_{n_{\eta,\lambda}}\rangle \rangle
\otimes | I_{n_\lambda} \rangle\rangle\Big) \notag\\
= &
\bigoplus_{\eta}
|f_{\eta,G_0} (g)\rangle \rangle \otimes
|X_\eta\rangle\rangle,
\end{align}
where we define $|X_\eta\rangle\rangle:=
\Big(\bigoplus_{\lambda \in \hat{G}_f}
 |I_{n_{\eta,\lambda}}\rangle \rangle
\otimes | I_{n_\lambda} \rangle\rangle\Big)$.

The average with respect to the Haar measure $\mu_0$ is 
\begin{align}
\rho_{\mu_0}=
\bigoplus_{\eta}
d_{\eta,G_0}^{-1}
 I_{\eta,G_0} \otimes  I_{\eta,G_0} \otimes
 |X_\eta \rangle \rangle
\langle\langle X_\eta |.
\end{align}

\begin{theorem}[General solutions]\label{appendix:THM4}
The optimal type-II error for unitary subgroup hypothesis testing of unitaries draw from $\mu$ and $\mu_{0}$ is given by,
\begin{align}
e^{D_{\max} (\rho_{\mu_0}\|\rho_{\mu})}
=
\max_{\eta}
d_{\eta,G_0}^{-1}  \sum_{\lambda} d_\lambda n_{\eta,\lambda}. \label{BNX}
\end{align}
\end{theorem}

\begin{proof}
We use the notation
$%c_{\eta,\lambda}
%\sqrt{n_{\eta,\lambda} (n_\lambda d-d_\lambda)}
|v_{\eta,\lambda}\rangle:= 
\frac{1}{n_\lambda}
 |I_{n_{\eta,\lambda}}\rangle \rangle
\otimes | I_{n_\lambda} \rangle\rangle$.
%where $c_{\eta,\lambda}>0$ is the normalizing constant.
and choose an orthonormal basis of ${\cal U}_{\eta,G_0}$
by $\{ |u_{j,\eta}\rangle \}_j$.
We define the projection $P_{\eta} $
to the space spanned by $\{ |u_{j,\eta}\rangle|u_{j',\eta}\rangle  |v_{\eta,\lambda}\rangle \}_{\lambda,j,j'}$.
We have
\begin{align}
\rho_{\mu_0}=
\bigoplus_{\eta}
d_{\eta,G_0}^{-1} 
%|u_{j,\eta}\rangle\langle u_{j,\eta}| 
I_{\eta} \otimes I_{\eta} \otimes
\sum_{\lambda,\lambda'} |I_{n_{\eta,\lambda}}\rangle \rangle 
\langle \langle I_{n_{\eta,\lambda'}}|
\otimes 
| I_{n_\lambda} \rangle\rangle\langle\langle I_{n_{\lambda'}}| ,
\end{align}
and,
\begin{align}
\rho_{\mu}
=&\int_{G}|f(g)\rangle\rangle \langle\langle f(g)|\mu(dg) \notag\\
=&\bigoplus_{\lambda \in \hat{G}_f}
d_\lambda^{-1} I_\lambda \otimes I_\lambda \otimes 
| I_{n_\lambda} \rangle\rangle\langle\langle I_{n_\lambda}| \notag\\
=&
\bigoplus_{\lambda}
d_\lambda^{-1}
(\bigoplus_{\eta}  I_\eta \otimes I_{n_{\eta,\lambda}}) 
\otimes
(\bigoplus_{\eta'}  I_{\eta'} \otimes I_{n_{\eta',\lambda}}) 
\otimes
| I_{n_\lambda} \rangle\rangle\langle\langle I_{n_\lambda}| \notag\\
=&
\bigoplus_{\lambda}
d_\lambda^{-1}
\bigoplus_{\eta,\eta'}  
I_\eta
\otimes
I_{\eta'}
\otimes 
I_{n_{\eta,\lambda}} \otimes I_{n_{\eta',\lambda}}
\otimes
| I_{n_\lambda} \rangle\rangle\langle\langle I_{n_\lambda}| 
\end{align}

$\rho_{\mu_0}$
and $\rho_{\mu}$ are commutative $ P_\eta$.
Also, 
we have
\begin{align}
\rho_{\mu_0}= 
\rho_{\mu_0}\Big(\sum_\eta P_\eta\Big).
\end{align}
Thus, we have
\begin{align}
e^{D_{\max} (\rho_{\mu_0}\|\rho_{\mu})}
=
\max_{\eta}
e^{D_{\max} (\rho_{\mu_0}P_\eta\|P_\eta\rho_{\mu})}.\label{VB1}
\end{align}
Since we have
\begin{align}
\rho_{\mu_0}P_\eta
=& d_{\eta,G_0}^{-1} 
I_{\eta} \otimes I_{\eta} \otimes
\sum_{\lambda,\lambda'} |I_{n_{\eta,\lambda}}\rangle \rangle \langle \langle I_{n_{\eta,\lambda'}}|
\otimes 
| I_{n_\lambda} \rangle\rangle\langle\langle I_{n_{\lambda'}}| 
\\
\rho_{\mu}P_\eta
=&
I_\eta
\otimes
I_{\eta}
\otimes 
\bigoplus_{\lambda}
d_\lambda^{-1}
I_{n_{\eta,\lambda}} \otimes I_{n_{\eta,\lambda}}
\otimes
| I_{n_\lambda} \rangle\rangle\langle\langle I_{n_\lambda}| ,
\end{align}
we have
\begin{align}
e^{D_{\max} (\rho_{\mu_0}P_\eta\|P_\eta\rho_{\mu})}
=
e^{D_{\max} (S_{\eta,1} \| S_{\eta,2} )}\label{VB2}
\end{align}
where
\begin{align}
S_{\eta,1}:=&d_{\eta,G_0}^{-1} 
\sum_{\lambda,\lambda'} |I_{n_{\eta,\lambda}}\rangle \rangle \langle \langle I_{n_{\eta,\lambda'}}|
\otimes 
| I_{n_\lambda} \rangle\rangle\langle\langle I_{n_{\lambda'} }| 
\\
S_{\eta,2}:=&
\bigoplus_{\lambda}
d_\lambda^{-1}
I_{n_{\eta,\lambda}} \otimes I_{n_{\eta,\lambda}}
\otimes
| I_{n_\lambda } \rangle\rangle\langle\langle I_{n_\lambda}| .
\end{align}
% \YA{
% Here we could use the 1-dimensional projector $\kett{X_\eta}\braa{X_\eta}$,
% \begin{align}
%     &e^{D_{\max} (S_{\eta,1} \| S_{\eta,2} )}= \frac{\tr[S_{\eta,1}^2]}{\tr[S_{\eta,1}S_{\eta,2}]}
%     =\frac{d_{\eta,G_0}^{-2}\braa{X_\eta}\!\!\kett{X_\eta}^2}{
%     \sum_{\lambda:\eta}d_{\eta,G_0}^{-1}d_{\lambda}^{-1}\braa{X_\eta}\!\!\kett{X_\eta}^2}\\
%     =&\frac{d_{\eta,G_0}^{-1}}{
%     n_{\eta,\lambda}d_{\lambda}^{-1}}
% \end{align}
% }
We define 
the projection $P_{2,\eta}$ to the space spanned by $\{|v_{\eta,\lambda}\rangle\}_\lambda$.
Since $S_{\eta,1}$ and 
$S_{\eta,2}$ are commutative with $P_{2,\eta}$, we have
\begin{align}
e^{D_{\max} (S_{\eta,1} \| S_{\eta,2} )}
=
e^{D_{\max} (S_{\eta,1} P_{2,\eta}\| S_{\eta,2} P_{2,\eta} )}.\label{VB3}
\end{align}
We choose a vector $|y\rangle:= 
\sum_{\lambda} y_\lambda |v_{\eta,\lambda}\rangle$.
Then, we have
\begin{align}
e^{D_{\max} (S_{\eta,1} P_{2,\eta}\| S_{\eta,2} P_{2,\eta} )}
=&\max_{|y\rangle}
\frac{\langle y| S_{\eta,1} P_{2,\eta} |y\rangle}
{\langle y| S_{\eta,2} P_{2,\eta} |y\rangle} \\
=& \max_{|y\rangle}
\frac{d_{\eta,G_0}^{-1} |\sum_{\lambda} y_\lambda n_{\eta,\lambda} |^2 }{
\sum_\lambda d_\lambda^{-1} n_{\eta,\lambda} |y_\lambda|^2} \\
=& d_{\eta,G_0}^{-1} \max_{|y\rangle}
\frac{|\sum_{\lambda} d_\lambda^{-1} n_{\eta,\lambda} d_\lambda y_\lambda  |^2 }{
\sum_\lambda d_\lambda^{-1} n_{\eta,\lambda} |y_\lambda|^2} \\
\stackrel{(a)}{=} & d_{\eta,G_0}^{-1} 
\max_{|y\rangle}
\frac{
\sum_{\lambda} d_\lambda^{-1} n_{\eta,\lambda} d_\lambda^2
\sum_\lambda d_\lambda^{-1} n_{\eta,\lambda} |y_\lambda|^2
}{
\sum_\lambda d_\lambda^{-1} n_{\eta,\lambda} |y_\lambda|^2} \\
= & d_{\eta,G_0}^{-1} 
\max_{|y\rangle}
\sum_{\lambda} d_\lambda n_{\eta,\lambda}
=  d_{\eta,G_0}^{-1}  \sum_{\lambda} d_\lambda n_{\eta,\lambda}.\label{VB4}
\end{align}
where $(a)$ follows from Schwarz inequality for inner product
$(x,y):= \sum_{\lambda} d_\lambda^{-1} n_{\eta,\lambda} \overline{x_\lambda} y_\lambda$
and choice $y_\lambda=d_\lambda$. Combining \eqref{VB1},
\eqref{VB2}, \eqref{VB3}, and \eqref{VB4},
we have,
\begin{align}
e^{D_{\max} (\rho_{\mu_0}\|\rho_{\mu})}
=
\max_{\eta}
d_{\eta,G_0}^{-1}  \sum_{\lambda} d_\lambda n_{\eta,\lambda}.
\end{align}

\end{proof}

\subsection{Proof of \eqref{KH2}: Parallel scheme}\label{sec:KH2}
To show \eqref{KH2}, we discuss the parallel scheme.
The aim of this section is to show
\begin{align}
&\min_{|\psi\rangle\in {\cal H}\otimes{\cal K}} 
\min_{T:inv} \{ 
\Tr T {\cal T}_G( |\psi\rangle\langle \psi|) : \min_{g \in G_0} \Tr T f(g) |\psi\rangle\langle \psi| f(g)^\dagger
\ge 1-\epsilon \} 
\le 
(1-\epsilon)\Big(\max_{\eta}
d_{\eta,G_0}^{-1}  \sum_{\lambda} d_\lambda n_{\eta,\lambda} \Big)^{-1}.\label{NMNM}
\end{align}
The combination of \eqref{NMNM} and \eqref{BNX} shows \eqref{KH2}.
The optimal performance of 
the parallel scheme is the LHS of the above relation.

When a reference system $\mathbb{C}^{l}$ is available, we have
\begin{align}
{\cal H}\otimes \mathbb{C}^{l}
= \bigoplus_{\lambda \in \hat{G}_f}
{\cal U}_\lambda \otimes \mathbb{C}^{l n_\lambda}.\label{NMU}
\end{align}
However, when the input state is a pure state, 
the orbit is restricted into the following space by choosing 
a suitable subspace $\mathbb{C}^{\min(d_\lambda,l n_\lambda)}$
of $\mathbb{C}^{l n_\lambda}$.

That is, our representation space can be considered as follows,
\begin{align}
\bigoplus_{\lambda \in \hat{G}_f}
{\cal U}_\lambda \otimes \mathbb{C}^{\min(d_\lambda,l n_\lambda)}.
\end{align}
In the following, we consider the above case.
We denote the projection to $
{\cal U}_\lambda \otimes \mathbb{C}^{\min(d_\lambda,l n_\lambda)}$
by $P_\lambda$.
When $n \ge d_\lambda/n_\lambda$ for any $\lambda \in \hat{G}_f$,
our representation is given as
\begin{align}
{\cal H}':=\bigoplus_{\lambda \in \hat{G}_f}
{\cal U}_\lambda \otimes \mathbb{C}^{d_\lambda}.
\end{align}

An invariant parallel strategy is given as follows.
We choose an input state $|\psi\rangle=
\oplus_{\lambda \in \hat{G}_f}\sqrt{p_\lambda} 
|\psi_\lambda\rangle$.
We choose a positive-semidefinite matrix $T$ on ${\cal H}'$ with $I \ge T\ge 0$.
Also, $T$ needs to satisfy the invariance condition for $G_0$.

Since
\begin{align}
{\cal U}_\lambda=
\bigoplus_{\eta} {\cal U}_{\eta,G_0} \otimes \mathbb{C}^{n_{\eta,\lambda}},
\end{align}
we have
\begin{align}
{\cal H}':=
\bigoplus_{\lambda \in \hat{G}_f}
\bigoplus_{\eta} {\cal U}_{\eta,G_0} \otimes \mathbb{C}^{n_{\eta,\lambda}}
\otimes 
(\bigoplus_{\eta'} \mathbb{C}^{d_{\eta',G_0}} \otimes \mathbb{C}^{n_{\eta',\lambda}}).
\end{align}

Hence, considering the case when $\eta=\eta'$,
we have a subspace
\begin{align}
\bigoplus_{\eta} 
{\cal U}_{\eta,G_0} \otimes \mathbb{C}^{d_{\eta,G_0}} 
\otimes 
\Big(\bigoplus_{\lambda \in \hat{G}_f}
\mathbb{C}^{n_{\eta,\lambda}}
\otimes \mathbb{C}^{n_{\eta,\lambda}}\Big)
\subset {\cal H}'.
\end{align}
Given $\eta$, we focus on the subspace 
\begin{align}
&
{\cal U}_{\eta,G_0} \otimes \mathbb{C}^{d_{\eta,G_0}} 
\otimes 
\Big(\bigoplus_{\lambda \in \hat{G}_f}
\mathbb{C}^{n_{\eta,\lambda}}
\otimes \mathbb{C}^{n_{\eta,\lambda}}\Big)
=
\Big(\bigoplus_{\lambda \in \hat{G}_f}
{\cal U}_{\eta,G_0} 
\otimes 
\mathbb{C}^{n_{\eta,\lambda}}
\otimes \mathbb{C}^{d_{\eta,G_0}} 
\otimes \mathbb{C}^{n_{\eta,\lambda}}\Big) \subset {\cal H}'.
\label{NBBI}
\end{align}
Here, $\mathbb{C}^{d_{\eta,G_0}} 
\otimes \mathbb{C}^{n_{\eta,\lambda}}$
is a subspace of $\mathbb{C}^{l n_\lambda}$ in \eqref{NMU}.

We set the distribution on $\hat{G}_f$ as
\begin{align}
p_\lambda:= \frac{d_\lambda n_{\eta,\lambda}}{
\sum_{\lambda' \in \hat{G}_f}d_{\lambda'} n_{\eta,\lambda'}}.
\end{align}
We define the vector 
\begin{align}
|F\rangle:= 
\sum_{\lambda \in \hat{G}_f}
p_\lambda^{1/2} n_{\eta,\lambda}^{-1/2}|I_{n_{\eta,\lambda}}\rangle\rangle
\in \Big(\bigoplus_{\lambda \in \hat{G}_f}
\mathbb{C}^{n_{\eta,\lambda}}
\otimes \mathbb{C}^{n_{\eta,\lambda}}\Big)\label{BNG7}
\end{align}
Using a vector 
\begin{align}
|f_{\eta,G_0}(e)\rangle\rangle
\in {\cal U}_{\eta,G_0} \otimes \mathbb{C}^{d_{\eta,G_0}} ,
\label{BNG8}
\end{align}
we define the vector
\begin{align}
|\psi\rangle:=
d_{\eta,G_0}^{-1/2}
|f_{\eta,G_0}(e)\rangle\rangle
\otimes |F\rangle.
\end{align}
Since the relations \eqref{NBBI}, \eqref{BNG7}, and \eqref{BNG8} guarantee the relation
$|\psi\rangle\in {\cal H}'$,
we can use $|\psi\rangle$ as our input state.

We define the decision operator
\begin{align}
T_0:=  I_{\eta,G_0}\otimes I_{d_{\eta,G_0}} \otimes |F\rangle \langle F|,
\end{align}
which satisfies $0 \le T_0 \le I$ and the invariance condition for $G_0$.

For $g \in G_0$, we have
\begin{align}
& \Tr f(g)|\psi\rangle\langle \psi|f(g)^\dagger T_0
=
\Tr |\psi\rangle\langle \psi|f(g)^\dagger I_{\eta,G_0}\otimes |F\rangle \langle F| f(g) 
\notag\\
=&
\Tr |\psi\rangle\langle \psi| I_{\eta,G_0}\otimes |F\rangle \langle F|  
=
\langle \psi| I_{\eta,G_0}\otimes |F\rangle \langle F|   |\psi\rangle\notag\\
=& 
d_{\eta,G_0}^{-1}
\langle f_{\eta,G_0}(e)| I_{\eta,G_0} \otimes I_{d_{\eta,G_0}} |f_{\eta,G_0}(e)\rangle
=1.
\end{align}
Also, we have
\begin{align}
{\cal T}_G(|\psi\rangle \langle \psi|)
=\bigoplus_{\lambda} p_\lambda d_\lambda^{-1} I_{\lambda} \otimes   
d_{\eta,G_0}^{-1} I_{d_{\eta,G_0}} \otimes n_{\eta,\lambda}^{-1} I_{n_{\eta,\lambda}} .
\end{align}
We have
\begin{align}
&\Tr {\cal T}_G(|\psi\rangle \langle \psi|) T_0 \notag\\
=& \Tr
\Big(\bigoplus_{\lambda\in \hat{G}_f} p_\lambda d_\lambda^{-1} I_{\lambda} \otimes   
d_{\eta,G_0}^{-1} I_{\eta,G_0} \otimes n_{\eta,\lambda}^{-1} I_{n_{\eta,\lambda}} \Big)
( I_{\eta,G_0}\otimes I_{d_{\eta,G_0}}
 \otimes |F\rangle \langle F| )\notag\\
=&
\sum_{\lambda\in \hat{G}_f} p_\lambda d_\lambda^{-1} 
d_{\eta,G_0}^{-1} n_{\eta,\lambda}^{-1}
\Tr [(I_{\lambda} \otimes \otimes  I_{n_{\eta,\lambda}} )
( I_{\eta,G_0} \otimes |F\rangle \langle F|)]
 \Tr I_{d_{\eta,G_0}} \notag\\
 =&
\sum_{\lambda\in \hat{G}_f} p_\lambda d_\lambda^{-1} 
 n_{\eta,\lambda}^{-1}
d_{\eta,G_0} p_\lambda \notag\\
 =&
d_{\eta,G_0}  \sum_{\lambda \in \hat{G}_f} p_\lambda^2 d_\lambda^{-1} 
n_{\eta,\lambda}^{-1} 
=d_{\eta,G_0} \Big(\sum_{\lambda\in \hat{G}_f}d_{\lambda} n_{\eta,\lambda}\Big)^{-1}.
\end{align}
Therefore, 
choosing the optimal $\eta$ in the sense of \eqref{BNX} and 
setting $T=(1-\epsilon)T_0$,
we obtain \eqref{NMNM}. % by using \eqref{BNX}.

Now, we consider the relation
\begin{align}\label{NMIY}
d_{\eta,G_0} n_{\eta,\lambda} \le n_\lambda
\end{align}
holds for an element $\lambda \in \hat{G}_f$ and an element $\eta \in \hat{G}_0$.
Then, instead of \eqref{NBBI}, we have the following relation,
\begin{align}
&
{\cal U}_{\eta,G_0} \otimes \mathbb{C}^{d_{\eta,G_0}} 
\otimes 
\Big(\bigoplus_{\lambda \in \hat{G}_f:\eqref{NMIY} \hbox{ holds.}}
\mathbb{C}^{n_{\eta,\lambda}}
\otimes \mathbb{C}^{n_{\eta,\lambda}}\Big)
=
\Big(\bigoplus_{\lambda \in \hat{G}_f:\eqref{NMIY} \hbox{ holds.}}
{\cal U}_{\eta,G_0} 
\otimes 
\mathbb{C}^{n_{\eta,\lambda}}
\otimes \mathbb{C}^{d_{\eta,G_0}} 
\otimes \mathbb{C}^{n_{\eta,\lambda}}\Big) \subset {\cal H}.\label{NBBI2}
\end{align}
Hence, when we restrict the range of $\lambda \in \hat{G}_f$
in the above protocol
to the set $\{\lambda \in \hat{G}_f:\eqref{NMIY} \hbox{ holds.}\}$,
the above protocol can be implemented without 
any reference system.
Thus, we have the following theorem,
\begin{theorem}\label{TH4}
For unitary subgroup hypothesis testing involving compact group $G$ and its subgroup $G_0$ with unitary representation $f$, an error tolerance of $\epsilon$ and the irreducible space $U_\lambda$ corresponding to $f$ has the constraint \eqref{NMIY}, we have,
\begin{align}
\min_{|\psi\rangle \in {\cal H}} \min_{T:inv} \{ 
\Tr T {\cal T}_G( |\psi\rangle\langle \psi|) :
\min_{g \in G_0} \Tr T f(g) |\psi\rangle\langle \psi| f(g)^\dagger 
\ge  1-\epsilon \}
\le
(1-\epsilon)\Big(\max_{\eta}
d_{\eta,G_0}^{-1}  \sum_{\lambda:\eqref{NMIY} \hbox{ holds.}} d_\lambda n_{\eta,\lambda} \Big)^{-1}.
\end{align}
\end{theorem}
As a special case, we have the following theorem.
\begin{theorem}\label{TH3}
Let $\eta_0$ be the optimal element in the sense of \eqref{BNX}.
When the relation \eqref{NMIY} 
holds for any $\lambda \in \hat{G}_f$ and the element $\eta=\eta_0 \in \hat{G}_0$,
the optimal test in Theorem \ref{NMI}
can be implemented by a parallel strategy without 
any reference system.
\end{theorem}

\section{Applications via case study} \label{sec:appendix case study}
\subsection{Application: Identity Test via studying the case with $G_0=\{I\}$ and $G$=SU($d$)}
Next, we consider the case when $G=$SU($d$) and $G_0=\{I\}$.
We consider the $n$-fold tensor product representation $f_n(g):=g^{\otimes n}$. According to Theorem~\ref{appendix:THM4}, when $G_0=\{I\}$, we have only one one-dimensional irreducible representation of $G_0$.
Hence, we have $n_{\eta,\lambda}=d_{\lambda}$ so that
\begin{align}
e^{D_{\max} (\rho_{\mu_0}\|\rho_{\mu})}
=
\max_{\eta}
d_{\eta,G_0}^{-1}  \sum_{\lambda} d_\lambda n_{\eta,\lambda}
=
\sum_{\lambda} d_\lambda^2
\end{align}

When $d=2$, we have the following calculations \cite{Population}.
When $n=2m$, 
the label $\lambda \in \hat{G}_{f_{2m}}$ is given as 
$j=0, \ldots, m$ so that we have 
\begin{align}
\sum_{\lambda \in \hat{G}_{f_{2m}}} d_\lambda^2
=\sum_{j=0}^{m} (2j+1)^2=\frac{(m+1)(2m+1)(2m+3)}{3}.
\end{align}
In this case, the multiplicity $n_j$
is ${2m \choose m-j}-{2m \choose m-j-1}$
for $j=0, \ldots,m-1$, and $n_m$ is $1$.
Since
\begin{align}
d_j=2j+1 \le n_j ={2m \choose m-j}-{2m \choose m-j-1}
\end{align}
for $j=0, \ldots,m-1$,
RHS of \eqref{NMNM} is calculated to
$(1-\epsilon)(\sum_{j=0}^{m-1} (2j+1)^2)^{-1}$.
That is, this value can be attained without any reference system.
This kind of observation has been done 
for the case of 
unitary estimation in \cite{CDS,H06}.
The reference \cite{H06} introduced the word ``self-entanglement'' 
to characterize this phenomena.
%n-2x+1=2j+1
%n+1-2j-1=2x
%2(m-j)
When $n=2m-1$, we have
\begin{align}
\sum_{\lambda \in \hat{G}_{f_{2m-1}}} d_\lambda^2
=\sum_{j=1}^{m} (2j)^2=\frac{2m(m+1)(2m+1)}{3}.
\end{align}
In this case, the multiplicity $n_j$
is ${2m-1 \choose m-j}-{2m-1 \choose m-j-1}$
for $j=0, \ldots,m-1$, and $n_m$ is $1$.
Since
\begin{align}
d_j=2j \le n_j ={2m-1 \choose m-j}-{2m-1 \choose m-j-1}
\end{align}
for $j=0, \ldots,m-1$,
RHS of \eqref{NMNM} is calculated to
$(1-\epsilon)(\sum_{j=0}^{m-1} (2j)^2)^{-1}$.
% That is, this value can be attained without any reference system.

%n-2x+1=2j
%n+1-2j=2x
%2(m-j)

% First, we fix $d$. 
We  then discuss the case when $n$ increases and we fix $d$.
We denote the set of Young inteces with length $n$ and depth $d$ by
$Y_n^d:=\{(\lambda_j)_{j=1}^d|
\sum_{j=1}^d\lambda_j=d, \lambda_j\ge \lambda_{j-1}
\}$ \cite{Group}, which equals $\hat{G}_{f_n}$.
Using \cite{Population}, we can show 
\begin{align}
\lim_{n\to \infty}\frac{\log (\sum_{\lambda \in Y_n^d} d_\lambda^2)}{\log n}
=d^2-1. \label{BNA} %=\frac{d(d-1)}{2}\log n +o(\log n).
\end{align}
Since $ d_\lambda \le (n+1)^{\frac{d(d-1)}{2}}$ and
$|Y_n^d| \le (n+1)^{d-1}$, we have $\le $ in \eqref{BNA}.
The opposite inequality can be shown by using the relation
\begin{align}
d_\lambda= \prod_{1\le i<j\le d} \frac{j-i+\lambda_j-\lambda_i}{j-i}.
\end{align}

\subsection{Application: Z-symmetry Test via studying the case with commutative $G_0$ and $G$=SU($d$)}
Next, we assume that $G=$SU($d$) and
$G_0$ is a maximal commutative subgroup of $G$.
We consider the $n$-fold tensor product representation $f_n(g):=g^{\otimes n}$.
Hence, $d_{\eta,G_0}=1$.
When $d=2$, 
any weight uniquely identify an eigen vector in ${\cal U}_\lambda$.
Theorem \ref{TH3} guarantees that 
the optimal strategy does not need the reference system because 
$n_{\eta,\lambda}=1$.
When $n=2k$, we have
\begin{align}
\max_{\eta \in \cup_{\lambda \in \hat{G}_f} d_{\eta,G_0}^{-1}
\hat{G}_{0,\lambda}}
\sum_{ \lambda \in \hat{G}_0}
d_\lambda  n_{\eta,\lambda}
=&\max_{\eta \in \cup_{\lambda \in \hat{G}_f}
\hat{G}_{0,\lambda}}
\sum_{ \lambda \in \hat{G}_0}
d_\lambda n_{\eta,\lambda}\\
=&\sum_{j=0}^{k} 2j+1 
=k(k+1)+k+1=(k+1)^2=\frac{(n+2)^2}{4}.
\end{align}
When $n=2k-1$, we have
\begin{align}
\max_{\eta \in \cup_{\lambda \in \hat{G}_f}d_{\eta,G_0}^{-1}
\hat{G}_{0,\lambda}}
\sum_{ \lambda \in \hat{G}_0}
d_\lambda n_{\eta,\lambda}
=\max_{\eta \in \cup_{\lambda \in \hat{G}_f}
\hat{G}_{0,\lambda}}
\sum_{ \lambda \in \hat{G}_0}
d_\lambda n_{\eta,\lambda}
=\sum_{j=1}^{k} 2j 
=k(k+1)=\frac{(n+1)(n+3)}{4}.
\end{align}
% First, we fix $d$. 
% We discuss the case when $n$ increases.
We note that for calculating the type-II error in the case of arbitrary dimensions, the following question need to be addressed:
\begin{quote}
Does the relation 
$n_{\eta,\lambda} \le n_{\lambda}$ hold
for any $\eta\in\hat{G}_{0,\lambda}, \lambda\in\hat{G}_{f} $?
\end{quote}
If the relation holds, Theorem \ref{TH3} guarantees that the optimal strategy does not need the reference system. Furthermore, we may also show,
\begin{align}
&\lim_{n\to \infty}
\frac{\log \max_{\eta \in \cup_{\lambda \in \hat{G}_f}
\hat{G}_{0,\lambda}}
\sum_{ \lambda \in \hat{G}_0}
d_\lambda n_{\eta,\lambda}}{\log n}
=
\frac{d(d-1)}{2}+\frac{d(d-1)}{2(d-1)}+d-1=
\frac{d^2+2d-2}{2}.
\end{align}

\subsection{Application: T-symmetry Test via studying the case with $G_0=O(d)$ and $G$=U($d$)}\label{Appendix:t-symmetry}
Next, we consider the case with $G=$U($d$) and $G_0=$O($d$) and the $n$-fold tensor product representation $f_n(g):=g^{\otimes n}$. It is noted that U($d$) has the same irreducible spaces as SU($d$). 

First, we consider the case with $d=2$, i.e., 
the space $\mathbb{C}$ spanned by $|e_0\rangle$ and $|e_1\rangle$,
where $|e_j\rangle:=\frac{1}{\sqrt{2}}(|0\rangle+(-1)^j i|1\rangle)$.
O($2$) has different irreducible representation spaces
from the irreducible representations of the commutative group.
O($2$) is equivalent to $U(1)\times \mathbb{Z}_2$.
For $(e^{i \theta},a) \in U(1)\times \mathbb{Z}_2$, we have
\begin{align}
(e^{i \theta},a)|e_j\rangle=
(e^{i (-1)^j\theta},a)|e_{j\oplus a}\rangle.
\end{align}
%Each irreducible representation of the commutative group
%is given as the one-dimensional space spanned by a weight vector.
%Each irreducible representations of O($2$) is given as 
%the space spanned by a weight vector and its opposite 
%weight vector.

We consider the case with $n=2k-1$.
All irreducible spaces of O($2$) are two-dimensional.
When $\eta$ is the irreducible representation with weights 
$|e_0\rangle^{\otimes l}|e_1\rangle^{\otimes n-l} $ 
and $|e_0\rangle^{\otimes n-l}|e_1\rangle^{\otimes l} $ with $l \le k-1$,
the dimension of 
the irreducible representation space of U($2$) 
to contain this irreducible space of O($2$) is 
 from $n-2l+1=2k-2l$ to $n+1=2k$.
 Therefore, 
\begin{align}
d_{\eta,G_0}^{-1}  \sum_{\lambda} d_\lambda n_{\eta,\lambda}
=\frac{1}{2}\sum_{j=k-l}^{k}2j
=\sum_{j=k-l}^{k}j=\frac{(l+1)(2k-l)}{2}.
\end{align}
 Then, we have
\begin{align}
&\max_\eta d_{\eta,G_0}^{-1}  \sum_{\lambda} d_\lambda n_{\eta,\lambda}
=\max_l \frac{(l+1)(2k-l)}{2}=\frac{k(k+1)}{2}
=\frac{(n+1)(n+3)}{8} .
\end{align}
 
Next, we consider the case with $n=2k$.
When $\eta$ is the irreducible representation with weights 
$|e_0\rangle^{\otimes l}|e_1\rangle^{\otimes n-l} $ 
and $|e_0\rangle^{\otimes n-l}|e_1\rangle^{\otimes l} $ with $l \le k-1$,
the irreducible space of O($2$) is two-dimensional
and the dimension of 
the irreducible representation space of U($2$)
to contain this irreducible space of O($2$) is 
 from $n-2l+1=2k-2l+1$ to $n+1=2k+1$.
 Therefore, 
\begin{align}
&d_{\eta,G_0}^{-1}  \sum_{\lambda} d_\lambda n_{\eta,\lambda}
=\frac{1}{2}\sum_{j=k-l}^{k}2j+1 
=\Big(\sum_{j=k-l}^{k}j\Big)+\frac{l+1}{2}
=\frac{(l+1)(2k-l)}{2}+\frac{l+1}{2}=\frac{(l+1)(2k-l+1)}{2}.
\end{align}
When $\eta$ is the irreducible representation with weights 
$|0\rangle^{\otimes k}|1\rangle^{\otimes k} $,
there are two types of one-dimensional irreducible spaces of O($2$).
One is 
$|e_0\rangle^{\otimes k}|e_1\rangle^{\otimes k} +|e_1\rangle^{\otimes k}|e_0\rangle^{\otimes k} $,
the other is 
$|e_0\rangle^{\otimes k}|e_1\rangle^{\otimes k} -|e_1\rangle^{\otimes k}|e_0\rangle^{\otimes k} $.
The dimension of the irreducible representation space of U($2$)
to contain the first-type irreducible space of O($2$) is 
$2k+1, 2k-3, \ldots, 3$ or $1$.
The dimension of the irreducible representation space of U($2$)
to contain the second-type irreducible space of O($2$) is 
$2k-1, 2k-5, \ldots, 3$ or $1$.
For the first-type of irreducible space, we have
 \begin{align}
&d_{\eta,G_0}^{-1}  \sum_{\lambda} d_\lambda n_{\eta,\lambda}
=(2k+1)+(2k-3)+  \ldots+ 3 \hbox{ or } 1=
\frac{(k+1)(k+2)}{2}.
%\sum_{j=0}^{k}2j+1 =k(k+1)+k+1=(k+1)^2
\end{align}
For the second-type of irreducible space, we have
 \begin{align}
&d_{\eta,G_0}^{-1}  \sum_{\lambda} d_\lambda n_{\eta,\lambda}
=(2k-1)+(2k-5)+  \ldots+ 3 \hbox{ or } 1<\frac{(k+1)(k+2)}{2}.
%\sum_{j=0}^{k}2j+1 =k(k+1)+k+1=(k+1)^2
\end{align}
Then, we have
\begin{align}
&\max_\eta d_{\eta,G_0}^{-1}  \sum_{\lambda} d_\lambda n_{\eta,\lambda}
=\max\Big(\frac{(k+1)(k+2)}{2},\max_l \frac{(l+1)(2k-l+1)}{2}\Big) 
=\max\Big(\frac{(k+1)(k+2)}{2},\frac{k(k+2)}{2}\Big)
=\frac{(k+1)(k+2)}{2}.
\end{align}

\end{document}